\documentclass [12pt]{article}

\usepackage{amsmath,amsthm,amscd,amssymb}
\def\alp{\alpha}

\def\ga{\gamma}
\def\de{\delta}
\def\ep{\varepsilon}

\def\la{\lambda}

\def\si{\sigma}

\def\De{\Delta}

\def\Om{\Omega}

\def\Th{\Theta}

\usepackage[small, centerlast, it]{caption}
\usepackage{epsfig}
\usepackage[titles]{tocloft}
\usepackage[latin1]{inputenc}
\usepackage{verbatim} %for comments via \begin{comment}
\usepackage[active]{srcltx} %inverse search
\usepackage{fancyhdr}
\usepackage{caption}

\def\b1{{1\!\!1}}

\def\cO{{\ca O}}
\def\cP{{\ca P}}

\def\bC{{\mathbb C}}           %%%  complex numbers and so on

\def\bI{{\mathbb I}}

\def\bR{{\mathbb R}}

\def\beq{\begin{eqnarray}}
\def\eeq{\end{eqnarray}}
\def\pa{\partial}
\newcommand{\ca}[1]{{\cal #1}}         %%  calligraphic

\newcommand{\mydef}{:=}

\newcommand{\wick}[1]{:\!{#1}\!:}

\newcommand{\fun}[1]{{\phantom{#1}}}

\newcommand{\Tr}{\text{\it Tr }}

\newcommand{\tensor}{\otimes}

\usepackage{sectsty}
\sectionfont{\normalsize}%\normalfont
\subsectionfont{\normalsize\normalfont\itshape}
\newtheoremstyle{thm}
{12pt}% space above
{12pt}% space below
{\itshape}% body font
{}% h indent amount
{\itshape\bfseries}% theorem head font
{}% punctuation after theorem head
{1em}% space after theorem head
{}% theorem head spec (can be left empty, meaning `normal')

\theoremstyle{thm}
\newtheorem{theorem}{Theorem}
\newtheorem{lemma}[theorem]{Lemma}

\newtheorem{definition}[theorem]{Definition}

\def\myem #1{\textbf{#1}}

\begin{document}

%%%%%%%%%%%%%   Title %%%%%%%%%%%%%%%%%%%%%%%%%%

\par
\bigskip
\large
\noindent
{\bf On the Stress-Energy Tensor of Quantum Fields in Curved Spacetimes - Comparison of Different Regularization Schemes and Symmetry of the Hadamard/Seeley-DeWitt Coefficients 
 }
\bigskip
\par
\rm
\normalsize

%%%%%%%%%%%%%%%%%%%%%%%%%%%%%%%%%%%%%%%%%%%%%
%%%%%%%%%%%% Authors %%%%%%%%%%%%%%%%%%%%%%%%

\noindent {\bf Thomas-Paul Hack$^{1,a}$} and {\bf Valter Moretti$^{2,b}$}\\
\par
\small
\noindent $^1$
II. Institut f\"ur Theoretische Physik, Universit\"at Hamburg,
Luruper Chaussee 149,
D-22761 Hamburg, Germany.\smallskip

\noindent $^2$
 Department of  Mathematics, Faculty of Science, University of Trento, via Sommarive 14, 38050 Povo (Trento), Italy.\smallskip

\noindent $^a$  thomas-paul.hack@desy.de, $^b$ moretti@science.unitn.it\\
 \normalsize

\par

\rm\normalsize

\noindent {\small Version of \today}

%\linespread{1.5}
\rm\normalsize

%%%%%%%%%%%% Date %%%%%%%%%%%%%%%%%%%%%%%%%%

\par
\bigskip

\noindent
\small
{\bf Abstract}.
We review a few rigorous and partly unpublished results on the regularisation of the stress-energy in quantum field theory on curved spacetimes: 1) the symmetry of the Hadamard/Seeley-DeWitt coefficients in smooth Riemannian and Lorentzian spacetimes 2) the equivalence of the local $\zeta$-function and the Hadamard-point-splitting procedure in smooth static spacetimes 3) the equivalence of the DeWitt-Schwinger- and the Hadamard-point-splitting procedure in smooth Riemannian and Lorentzian spacetimes.
\normalsize

%\tableofcontents

\section{Introduction}

The ``semiclassical'' Einstein equation
\beq\label{see}G_{\mu\nu}=8\pi G\langle \wick{T_{\mu\nu}}\rangle_\Om\,,\eeq
is obtained from the ``classical'' Einstein equation by replacing the classical stress-energy tensor $T_{\mu\nu}$ with the expectation value of the quantum stress energy tensor $\wick{T_{\mu\nu}}$ in a quantum state $\Om$. There are several angles from which one can approach the question whether this equation is sensible and well-defined. On the one hand, one can argue that it can be obtained from a theory of (perturbative) quantum gravity in a suitable limit where gravity becomes classical, but matter is still quantized, hence the name; see e.g. \cite{FW96} for a review of this issue. On the other hand, if we combine our understanding of e.g. radiation being described by quantum field theory on the fundamental level and the fact that radiation does have an influence on spacetime curvature for all we know from cosmological observations, we can safely claim that the latter observations indicate that (something like) the semiclassical Einstein equation {\it must} hold at least in a suitable regime, although it is certainly not a fundamental equation, as it equates a deterministic classical quantity with a probabilistic quantum expression.

Even if one accepts one or both of the aforementioned physical motivations to consider the semiclassical Einstein equation, questions about the well-posedness of this equation remain. We usually assume that spacetime is regular outside of black hole or big bang singularities, so $\langle \wick{T_{\mu\nu}}\rangle_\Om$ should better be regular as well. However, as the classical expression for $T_{\mu\nu}$ is (at least) quadratic in the field, one can expect that the naive quantum stress energy tensor is divergent and has to be regularised in some way. Indeed, various schemes to regularise the quantum stress-energy tensor have been proposed, including some more familiar from flat spacetime like dimensional regularisation \cite{Brown}, $\zeta$-function regularisation \cite{Dowker, Hawking} or Pauli-Villars regularisation \cite{Birrell} and also ones which are less familiar in that context like point-splitting regularisation \cite{Christensen0, Christensen, Wald2, Wald3}. Naturally, the question arose which regularisation scheme was the ``correct'' one, and a satisfactory answer was first given in \cite{Wald2} (and later generalised in \cite{HW04}). In \cite{Wald2}, Wald proposed a minimal set of well-motivated conditions which any sensible regularisation scheme should satisfy. We state these conditions, the ``Wald axioms'', here in their modern form \cite{WaldBook2, BFV, HW04}.
\begin{enumerate}
 \item The commutator of the regularised stress-energy tensor $\wick{T_{\mu\nu}(x)}$ with any product of fields $\phi(x_1)\cdots\phi(x_n)$ equals their  commutator with the non-regularised stress-energy tensor $T_{\mu\nu}(x)$\footnote{This commutator has to be defined properly, e.g. by performing a point-splitting of $T_{\mu\nu}(x)$, computing the commutator with this point-split operator and then taking the coinciding point limit of the result.}.
 \item $\wick{T_{\mu\nu}(x)}$ transforms covariantly under diffeomorphisms and does not depend on the metric and its derivatives at $y\neq x$.
  \item (The expectation value of) $\wick{T_{\mu\nu}(x)}$ has vanishing covariant divergence.
  \item In Minkowski spacetime and in the Minkowski vacuum state the expectation value of $\wick{T_{\mu\nu}(x)}$ vanishes.
 \item (The expectation value of) $\wick{T_{\mu\nu}(x)}$ does not contain derivatives of the metric of order higher than two.
\end{enumerate}
While the first condition makes sure that the regularisation scheme is ``minimal'' and ``state-independent'', as the divergence in the non-regularised stress-energy tensor is proportional to the identity, the second assures that is local and purely geometric, and the third follows from the fact that $G_{\mu\nu}$ has vanishing covariant divergence as well. The fourth condition, introduced to guarantee that the regularisation prescription coincides with ``normal ordering'' in flat spacetime, has to be omitted if one wants to have a non-vanishing cosmological constant (on the right hand side of the Einstein equations) \cite{FullingBook}. Finally, the last condition was meant to avoid the appearance of unstable, ``runaway'' solutions to the semiclassical Einstein equation. Although it turned out that this condition can not be fulfilled and that such unstable solutions can not be avoided, one can still analyse the subclass of stable solutions in a controlled manner in special cases, see e.g. \cite{Star, Vile85, DFP, Koksma, Nicola, DHMP}.

Wald has conjectured already in \cite{Wald79} that $\zeta$-function regularisation, dimensional regularisation and the various point-splitting regularisation schemes satisfy the essential axioms 1., 2. and 3. and thus all have a raison d'\^etre although they give slightly different results. In fact, one can infer directly from the axioms 1., 2. and 3. that their solution can not be unique, but is subject to a finite regularisation freedom which can be (under suitable assumptions) parametrised by 4 divergence-free tensors constructed out of the metric and the Riemann tensor \cite{Wald2, HW04}; in the absence of a fundamental theory of quantum gravity and matter, this freedom can only be fixed by comparison with experiments \cite{DFP, DHMP}.  So Wald could not only motivate why most regularisation prescriptions on the market should be legitimate but also explain why they naturally give different results. Nevertheless, the understanding of the renormalisation ambiguity still seems to be a topic of active research, see e.g. \cite{Shapiro}.

In this work we shall be concerned with proving parts of Wald's
conjecture, namely, the rigorous relation between three different
regularisation schemes, obtained in partially already published
\cite{MorettiZeta, Moretti2, Moretti} and partially yet unpublished
\cite{Thomas} previous works of the authors. On the one hand, we shall
review the equivalence between the local $\zeta$-function
regularisation and the Hadamard point-splitting regularisation in
Riemannian (``Euclidean'') manifolds and the compatibility of the
local $\zeta$-function scheme with Wick-rotation whenever
applicable. This scheme is only meaningful in Riemannian manifolds
because it is based on spectral properties of the field equations,
which strongly depend on the metric signature. Hence, a rigorous
comparison with the Hadamard scheme is only possible in the Riemannian
context {\it a priori} and can be taken over to the Lorentzian setting
by means of a Wick-rotation only in static Lorentzian
spacetimes. The situation is comparable in the second case we shall
discuss, the well-posedness of the DeWitt-Schwinger prescription. To
wit, like the local $\zeta$-function scheme and for the same reasons,
the DeWitt-Schwinger prescription is only well-defined in Riemannian
manifolds. Nevertheless, the computations of the regularised
stress-tensor for free quantum fields of arbitrary spin in Lorentzian
spacetimes in \cite{Christensen, ChristensenDuff}, obtained with the
DeWitt-Schwinger prescription, are now part of the standard literature
\cite{Birrell}, although it was clear that their derivation had not
been rigorous in Lorentzian spacetimes and that their legitimation via
a  Wick-rotation was not possible in general Lorentzian spacetimes
(see for instance the comments at the beginning of section 6.6 in
\cite{Birrell}). Moreover, as we shall explain in the main body of
this work, the DeWitt-Schwinger prescription to regularise the
expectation value of the stress-energy tensor fails to take into
account the state-dependence of this object and thus only gives the
``geometric piece'' of $\langle\wick{T_{\mu\nu}}\rangle_\Om$; this,
however, is sufficient to determine the so-called conformal anomaly
\cite{Duff, Wald2}. Hence, although the results of \cite{Christensen,
  ChristensenDuff} on the conformal anomaly have been confirmed with
rigorous calculations for Klein-Gordon \cite{Wald2, Mo03} and Dirac
\cite{DHP} fields by now, it is still important to know whether the
DeWitt-Schwinger scheme is generally legitimate on arbitrary Lorentzian spacetimes and if one can formulate it in a way which does take into account the full state dependence of $\langle\wick{T_{\mu\nu}}\rangle_\Om$; we shall prove both of these aspects in what follows. Although our proof is based on the direct relation between the so-called Seeley-DeWitt coefficients and the Hadamard coefficients, the DeWitt-Schwinger prescription still differs from the Hadamard scheme in that it does not consist of just applying a suitable differential operator to the regularised two-point function. Our results imply that the local $\zeta$-function and the DeWitt-Schwinger prescription both satisfy Wald's axioms just like the Hadamard regularisation scheme, though, in the $\zeta$-function case, only on Lorentzian manifolds where a Wick-rotation is possible.

As the Seeley-DeWitt/Hadamard coefficients are of such an utmost
importance in the point-splitting regularisation prescriptions (and in
non-commutative geometry \cite{Connes}), we complement our analysis of
the different regularisation schemes by reviewing a proof of their
symmetry. Although this symmetry had been heavily used and expected to
hold because these coefficients also appear in the symmetric Green's
functions of the field equations \cite{Friedlander, Garabedian}, a
rigorous proof was missing since the coefficients appear only in
infinite series whose convergence is unclear except in analytic spacetimes. The proof of \cite{Moretti2, Moretti} proceeds in three steps. First the symmetry is proven in analytic Riemannian manifolds by employing the symmetry of the field equations. Then, the proof is extended to analytic Lorentzian spacetimes by means of a {\it local} Wick rotation, i.e. one ``rotates'' the metric rather than the coordinates. Finally, the symmetry is established in general non-analytic Lorentzian spacetimes by locally approximating in a controlled way smooth metrics with analytic ones.

Our paper is organised as follows. After introducing a few preliminary
notions in section 2, we continue in section 3 by discussing the
expansion of the heat kernel in terms of the Seeley-DeWitt
coefficients and the expansion of the two-point function of quantum
field in terms of the Hadamard coefficients, followed by a proof of
their symmetry in section 4. In sections 5 and 6 we review the proof
of the equivalence of the local $\zeta$-function and the Hadamard
regularisation scheme in Riemannian manifolds and static Lorentzian manifolds and the well-posedness of the DeWitt-Schwinger regularisation prescription in Lorentzian spacetimes, respectively. The paper ends with our conclusions in section 7.

For simplicity and to increase the accessibility of this work, we consider only the real Klein-Gordon field in four dimensions. However, some of the results we review have been proven in arbitrary dimensions \cite{MorettiZeta, Moretti2, Moretti}, and we have no reason to doubt that they can be extended to fields of higher spin as well.\\

In the following a manifold (including bundles) is supposed to be smooth (i.e. $C^\infty$), unless specified otherwise, Hausdorff and second countable (to assure paracompactness).  
Vector fields, tensor fields, metrics and sections are supposed to be smooth as well unless specified otherwise.
We choose the signature $(-,+,+,+)$ for Lorentzian metrics and follow \cite{WaldBook} regarding the definitions of the Riemann and Ricci tensor.

 A spacetime is a time oriented four-dimensional manifold equipped with a smooth Lorentzian metric. In the text, Riemannian and Euclidean 
 are synonyms.

\section{Gaussian states, static spacetimes, Wick rotation}\label{secstatic}
Throughout we suppose to consider  globally hyperbolic spacetimes $(M,g)$ only and we moreover stick to the case of a real scalar field $\phi$ propagating in  $(M,g)$, 
satisfying:
\begin{eqnarray}
P \phi = 0 \label{evolution}\:,\quad  P := -\nabla_\mu\nabla^\mu + V = -\Box + V\:,
\end{eqnarray}
the covariant derivative $\nabla$ being always associated with the metric $g$ and
 $V$ being a smooth scalar 
function of the form
\begin{eqnarray}
V(x) := \xi R+m^2 +V'(x) \label{V}\:.
\end{eqnarray}
Above $\xi$ is a real constant,
 $R$ is the scalar 
curvature and $m\geq 0$ the mass of the particles associated to the field.
The domain of $P$ is the space of real-valued $C^\infty$ functions with compactly supported Cauchy data.
A {\it quasifree state} -- also known as {\it Gaussian state} -- $\Omega$ is unambiguously defined by choosing a Wightman two-point function $\langle\phi(x)\phi(y)\rangle_\Om$. The GNS theorem or the standard Wick expansion 
allows to construct a corresponding Fock space realisation of the theory.

 A globally hyperbolic spacetime $(M,g)$  is said to be
{\em static} if  it admits a time-like Killing vector field $X:= \partial_t$  normal to a smooth spacelike Cauchy surface 
$\Sigma$. Consequently, there are  (local) coordinate frames
$(x^0,x^1,x^2,x^3)\equiv (t,\vec{x})$ where $g_{0i} = 0$ ($i=1,2,3$) and $\partial_t g_{\mu\nu}=0$ and $\vec{x}$ are local coordinates on $\Sigma$.
When discussing the static case, we also assume that $V'$ satisfies $\partial_t V'=0$ so that the space of solutions of (\ref{evolution}) is invariant under $t$-displacements.

In a globally hyperbolic static spacetime it is possible to construct a $t$-invariant Gaussian state $\Omega$ by performing
 the {\em Wick rotation}, i.e. obtaining the Euclidean formulation of the same QFT.
This means that one can pass from the Lorentzian manifold
$(M,g)$ to an associated Riemannian manifold $(M_E, g^{(E)})$ by the analytic
continuation $t\to i\tau$ where $\tau \in \bR$.
In this way, $\partial_t$ gives rise to another Killing vector $\partial_\tau$ in $M_E$.  $M_E$ can alternatively be defined by assuming that the orbits of the Euclidean time $\tau$ are closed with period $\beta$, including the case $\tau \in \bR$
as $\beta= +\infty$.  Within this approach \cite{fr}, the two-point function of $\Omega$  is completely determined by 
 a proper Green's function (in the spectral theory
sense) $S(\tau-\tau',\vec{x},\vec{x}')$
of a corresponding  self-adjoint extension $P_E$ of the operator
\begin{eqnarray}
P_E' :=  -\nabla^{(E)}_a\nabla^{(E) a} + V(\vec{x}) = -\Delta^{(E)} + V(\vec{x}) \::\:  C_0^\infty(M_E) \to L^2(M_E,d_{g^{(E)}}x)\:.
\end{eqnarray}
$S(\tau-\tau',\vec{x},\vec{x}')$ is periodic with period $\beta \in (0,+\infty) \cup \{+\infty\}$ in the $\tau-\tau'$ entry and it is
 called the {\it Schwinger function} of $\Omega$.  $S$ is the integral kernel of
$P_E^{-1}$ when $M_E$ is compact and $P'_E>0$ \cite{Wald79} and, in that case, $P_E= \overline{P'_E}$, that is $P'_E$ is essentially self-adjoint. In general a $t$-invariant state $\Omega$ is determined by the choice of a self-adjoint extension
of the Wick rotated Klein-Gordon operator $P'_E$. 
When $0<\beta < +\infty$, the value $T=1/\beta$ has to be interpreted as 
the temperature of a state $\Om$ because the Wightman two-point function $\langle\phi(x)\phi(y)\rangle_{\Om}$ satisfies the KMS condition at the inverse temperature $\beta$.

\section{Hadamard / Heat-kernel Seeley-DeWitt coefficients}

In order to introduce the Hadamard and Seeley-DeWitt coefficients we need a few definitions and results from bitensor calculus.

If $VM$ and $WM$ are vector bundles over a manifold $M$ with typical fibers constituted by the vector spaces $V$ and $W$ respectively, then we denote by $VM\boxtimes WM$ the {\em exterior tensor product} of $VM$ and $WM$. $VM\boxtimes WM$ is defined as the vector bundle over $M\times M$ with typical fibre $V\tensor W$. The more familiar notion of the tensor product bundle $VM\tensor WM$ is obtained by considering the pull-back bundle of $VM\boxtimes WM$ with respect to the map $M\ni x\mapsto (x,x)\in M^2$. Typical exterior product bundles are for instance the tangent bundles of Cartesian products of $M$, {\it e.g.} $T^*M\boxtimes T^*M = T^*M^2$. A section of $VM\boxtimes WM$ is called a {\em bitensor}. We introduce the {\em Synge bracket notation} for the coinciding point limits of a bitensor. Namely, let $B$ be a smooth section of $VM\boxtimes WM$. We define
$$[B](x)\mydef B(x,x)\,.$$
With this definition, $[B]$ is a section of $VM\tensor WM$. In the following, we shall denote by unprimed indices tensorial quantities at $x$, while primed indices denote tensorial quantities at $y$. Moreover, covariant derivatives shall be denoted by the usual abbreviated notation 
$$B_{;\mu}\mydef\nabla_\mu B\;.$$
As an example, let us state the well-known {\em Synge rule}, proved for instance in \cite{Christensen0}.
\begin{lemma}
 \label{lem_SyngeRule} Let $B$ be an arbitrary smooth bitensor. Its covariant derivatives at $x$ and $y$ are related by \myem{Synge's rule}. Namely,
$$[B_{;\mu^\prime}]=[B]_{;\mu}-[B_{;\mu}]\,.$$
Particularly, let $VM$ be a vector bundle,  let $f_a$ be a local frame of $VM$ defined on $\cO\subset M$ and let $x$, $y\in\cO$. If $B$ is \myem{symmetric}, {\it i.e.} the coefficients $B_{ab^\prime}(x,y)$ of $$B(x,y)\mydef B^{ab^\prime}(x,y)\;f_a(x)\tensor f_{b^\prime}(y)$$ fulfil $$B^{ab^\prime}(x,y)=B^{b^\prime a}(y,x)\,,$$
then
$$[B_{;\mu^\prime}]=[B_{;\mu}]=\frac12 [B]_{;\mu}\,.$$
\end{lemma}
\noindent Other relevant properties of $[\cdot]$ can be found in  e.g. \cite{FullingBook}.

\subsection{Hadamard coefficients, Hadamard states and point-splitting regularisation}
\label{HadamardCoefficients}
An important tool of the Hadamard point-splitting renormalisation prescription in curved spacetime are the so-called {\em Hadamard coefficients}. They appear 
in the explicit expression of the {\em Hadamard parametrix}. That notion, in turn, has its physical reason in the definition of a relevant class of quantum states called (locally) {\em Hadamard states} \cite{Kay,Radzikowski,Radzikowski2,WaldBook2}. To discuss them, we consider a Gaussian  state $\Omega$ for a real scalar field $\phi$ on a globally  hyperbolic spacetime $(M,g)$ as before. 
\begin{definition}
 \label{def_HadamardFormScalar}
We say that $\langle\phi(x)\phi(y)\rangle_\Om$ is of \myem{local Hadamard form}, and $\Omega$ is {\bf locally Hadamard}, if:
\begin{align*}\langle\phi(x)\phi(y)\rangle_\Om&=\lim_{\ep\downarrow 0}\frac{1}{8\pi^2}\left(\frac{u(x,y)}{\si_\ep(x,y)}+\sum\limits^\infty_{n=0}v_n \si^n\log\left(\frac{\si_\ep(x,y)}{\la^2}\right)+w(x,y)\right)\\
&\mydef\frac{1}{8\pi^2}\left(h_\ep(x,y)+w(x,y)\right)\,,\end{align*}
where $x,y$ are taken in any geodesically convex neighbourhood. Above: 
\beq \sigma_\ep(x,y) := \sigma(x,y) +2i\ep (T(x)-T(y)) + \ep^2\label{sigma}\eeq 
where $T$ is any local time coordinate increasing toward the future, and $\sigma(x,y)$ denotes one half of 
 the signed squared geodesic distance of $x$ and $y$
(well-defined in a geodesically convex neighborhood) and $\lambda>0$ a renormalization length scale. Finally $w$ is a smooth 
biscalar determined by the state up to re-definition of $\lambda$.\\ 
 \indent The bidistribution $h_\ep$ is the \myem{Hadamard parametrix} and the coefficients (bi-scalars) $v_n$ are called
 \myem{Hadamard coefficients}. 
\end{definition}
\noindent The convergence of the series in the expression of $h_\ep$ is just asymptotic. However a proper point-wise convergence can be obtained by introducing a set of suitable cutoff functions (see \cite{Friedlander,HW01} for details).
The parametrix $h_\ep$ solves the Klein-Gordon equation
 in each argument separately up to smooth terms.
From that one obtains the following recursive differential equations (actually they do not depend on the signature of the metric)
already studied by Riesz in his seminal paper \cite{Riesz} (see also \cite{Garabedian}), determining $u=u(x,y)$ and all the Hadamard coefficients $v_n= v_n(x,y)$ in a fixed geodesically convex neighborhood when assuming $[u]=1$.
\begin{gather}
2u_{;\mu}\si_{;}^{\fun{;}\mu}+(\square_x\si-4)u=0\,,\label{eq_PDEUScalar}\\
-P_xu+2v_{0;\mu}\si_{;}^{\fun{;}\mu}+(\square_x\si-2)v_0=0\,,\label{eq_PDEV0Scalar}\\
-P_x v_n+2(n+1)v_{n+1;\mu}\si_{;}^{\fun{;}\mu}+\left(\left(n+1
\right)\square_x\si+2n\left(n+1\right)\right)(v_{n+1})=0\,\quad\forall n\geq 0\label{eq_PDEVNScalar}
\end{gather}
In particular, it turns out that $u$ coincides with the square root of the so called {\em VanVleck-Morette determinant} \cite{FullingBook, Moretti2}.
As a consequence, the following relations can be proven (see e.g. \cite{Mo03}):
\begin{lemma}
 \label{lem_PHadamardScalar}
The following identities hold for the Hadamard parametrix $h_\ep(x,y)$
$$[P_xh_\ep]=[P_yh_\ep]=-6[v_1]\,,\quad [(P_xh_\ep)_{;\mu}]=[(P_yh_\ep)_{;\mu^\prime}]=-4[v_1]_{;\mu}\,,$$$$[(P_xh_\ep)_{;\mu^\prime}]=[(P_yh_\ep)_{;\mu}]=-2[v_1]_{;\mu}\,.$$
\end{lemma}
\noindent These identities have a great deal of effect concerning the {\em point-splitting renormalisation} of the stress energy tensor evaluated 
on a Hadamard state $\Omega$. It is essentially defined by subtracting the universal Hadamard singularity from the two-point function of $\Omega$, before computing the relevant derivatives. 
\beq\label{set}
\left<\wick{T_{\mu\nu}}(x) \right>_{\Omega} = \left(D_{(x)\mu\nu}(x,y) \left(\left< \phi(x) \phi(y) \right>_{\Omega}-h_\ep(x,y)\right)\right)|_{x=y}
\label{EMTensor}
\eeq
\noindent $D_{\mu\nu}(x,y)$ is the following second order partial differential operator\footnote{Often, the
  symmetrised version of $D_{\mu\nu}$ is considered. However, 
  since $\left< \phi(x) \phi(y) \right>_{\Omega}-h_\ep(x,y)$ is
  already symmetric, it is not necessary to symmetrise the
  differential operator.}
  (cf. \cite{Mo03} eq. (10), and \cite{Thomas}, \cite{EG11} where some minor misprints have been corrected), 
\begin{align*}
D_{\mu\nu}(x,y) \mydef&\; D^\text{can}_{\mu\nu}(x,y) +\frac13 g_{\mu\nu} P_x\\
D^\text{can}_{\mu\nu}(x,y)\mydef&\; (1-2\xi)g^{\nu^\prime}_\nu\nabla_{\mu}\nabla_{\nu^\prime}-2\xi\nabla_{\mu}\nabla_\nu+\xi G_{\mu\nu}\\&\;+g_{\mu\nu}\left\{2\xi \square_x+\left(2\xi-\frac12\right)g^{\rho^\prime}_\rho\nabla^\rho\nabla_{\rho^\prime}-\frac12 m^2\right\} 
\end{align*}
Here covariant derivatives with primed indices indicate covariant derivatives w.r.t. $y$, $g^{\nu^\prime}_\nu$ denotes the parallel transport of vectors along the unique geodesic connecting $x$ and $y$, the metric $g_{\mu\nu}$ and the Einstein tensor $G_{\mu\nu}$ are considered to be evaluated at $x$ and we have assumed $V^\prime=0$ for simplicity.

The form of the ``canonical" piece $D^\text{can}_{\mu\nu}$ follows from the
definition of the classical stress-energy tensor, while the last term $\frac13 g_{\mu\nu} P_x$ has been introduced in \cite{Mo03} and it gives no contribution classically, just in view of the very Klein-Gordon equation satisfied by the fields.
However, in the quantum realm, its presence has important consequences due to  Lemma \ref{lem_PHadamardScalar} because the Hadamard parametrix satisfies the Klein-Gordon equation only up to smooth terms and thus the above definition without this additional term would not yield a conserved stress-tensor expectation value. As a matter of fact the following theorem is valid (\cite{Mo03} Theorem 2.1).
\begin{theorem}
 \label{thm_TensorScalar} The stress-energy tensor expectation value defined by (\ref{set}) fulfils the Wald axioms 1., 2., 3. for every value of $\lambda >0$ and, 
with a suitable choice for the scale $\lambda>0$, it satisfies 4. 
For all $\la$ it displays the trace anomaly when the conformal coupling $\xi =1/6$   and $m=0$ are chosen:
$$g^{\mu\nu}\left<\wick{T_{\mu\nu}}(x) \right>_{\Omega} = -\frac{v_1(x,x)}{4\pi^2}\:.$$
\end{theorem}
To conclude we remark an important theoretical result: if a globally hyperbolic spacetime is static and $\Omega$ is defined as in sec. \ref{secstatic} through the Wick rotation procedure, $\Omega$ turns out to be Hadamard \cite{WaldBook2,Mo03} for $0<\beta \leq +\infty$. In particular the Minkowski 
vacuum, that can be constructed in that way,  is Hadamard.
\subsection{Seeley-DeWitt coefficients}
\label{SeeleyDeWitt}
We now introduce another class of biscalars called the {\em Seeley-DeWitt coefficients}. They  appear if on tries to regularise both  the {\it Euclidean} and {\em Lorentzian one-loop effective action} by means of the asymptotic expansion of the {\it heat kernel}.
Actually the notion of heat kernel is properly defined only in the Euclidean case, though the terminology has been adopted for the Lorentzian case as well.
For the moment we focus on the mathematical issues only, and just from
a very formal point of view,  while we leave aside the physical
motivations leading to the Euclidean approach (especially the
definition of the considered quantum state) as well as precise
mathematical definitions that will be discussed in a subsequent
section. We just say that a formal motivation for passing from the
Lorentzian to Euclidean side and vice versa relies upon the natural
extension of the Wick rotation procedure discussed in sec. \ref{secstatic} when dealing with static spacetimes. 
 On a Euclidean manifold $(M,g)$, the {\em Euclidean effective action} $S_\text{eff}$ is defined via an Euclidean path integral as
$$e^{-S_\text{eff}}=\int [d\phi] e^{-\frac12 \langle \phi, P\phi\rangle}\,,$$
 \beq P := -\Delta + V : C_0^\infty (M) \to L^2(M, d_gx)\label{PE}\eeq
$\Delta$ being  the Laplace-Beltrami operator associated with  $g$ and $V: M \to \bR$ being some given smooth  function such that $P\geq 0$,  is the Euclidean Klein-Gordon operator defined on  $(M,g)$.
 $[d\phi]$ is supposed to denote some (formal) measure on the space of field configurations. The argument of the exponential in the integrand is the classical (Euclidean) Klein-Gordon action taking boundary terms into account. 
On Lorentzian manifolds, formally speaking, one would put an imaginary
unit in front of both actions in the above formula. In that case,
however, the rigorous heat-kernel theory cannot be applied since it
strongly relies upon the fact that $\Delta$ is (strongly) elliptic and
several regularity results (including convergences of expansions)
generally cease to hold true. See \cite{Wald79,FullingBook, Moretti2}
for comments on these issues.

The relevance of the effective action in our case is the fact that one can formally define the expectation value of the regularised stress-energy tensor in analogy to the classical case as
\beq \langle\wick{T_{\mu\nu}}\rangle_\Om\mydef\frac{2}{\sqrt{|\det g|}}\frac{\de S_\text{eff}}{\de g^{\mu\nu}}\,.\label{stressSeff}\eeq
The natural question which arises is where the state ($\Omega$)
dependence of the right-hand side is hidden. Indeed, one can interpret
the situation in such a way that the measure $[d\phi]$ is
state-dependent \cite{Mo03}. We shall consider these issues in
forthcoming sections. To compute $\langle\wick{T_{\mu\nu}}\rangle_\Om$ in
this picture, one therefore needs to compute the mentioned functional
derivative with respect to the metric. As we are going to see, already
$S_\text{eff}$  is affected by divergences and thus a regularization
process is necessary. If the regularisation is such that the effective
action is diffeomorphism-invariant, than the resulting functional
derivative will automatically lead to a conserved
$\langle\wick{T_{\mu\nu}}\rangle_\Om$.

Let us proceed to understand why the effective action is divergent and how to regularise it. Evaluating the above integral as if we were in a finite dimensional vector space, thus computing a Gaussian integral, one finds (up to a constant)
\beq e^{-S_\text{eff}}=(\det P)^{-\frac12}=e^{-\frac12 \Tr\log P}\,,\label{idtracedet}\eeq
and, hence,
$$S_\text{eff}=\frac12 \int\limits_M d_gx\; [\log P](x,x)\mydef  \int\limits_M d_gx\; L_\text{eff}(x)\,,$$
where the identity (valid in finite dimensional cases) $\log \det P=\Tr\log P$ has been used, the trace is evaluated by 
integrating the kernel of the operator $\log P$ at coinciding points, and $L_\text{eff}$ is interpreted as the {\em effective Lagrangean}. The above integral is certainly diverging, as one can already infer from the identity (that holds true for $M$ compact  when replacing  $P$ in (\ref{PE})
by its unique self-adjoint extension $\overline{P}$ and referring to the weak operator topology)
$$\log P=\lim_{\ep\downarrow 0}\left\{-\int\limits^\infty_\ep ds \;\frac{e^{-sP}}{s}+(\ga-\log \ep)\bI\right\}\,,$$
where $\ga$ denotes the Euler-Mascheroni constant and $\bI$ is the
identity operator. We discard the divergent term proportional to
$\bI$, that is trace class only in finite dimensional spaces, and
define (where again, to be rigorous, one should replace $P$ by
$\overline{P}$)
\beq\label{eq_Leff}L_\text{eff}(x)\mydef-\int^\infty_0ds\;\frac{\left[e^{-sP}\right](x,x)}{s}\,.\eeq
The integral kernel $[e^{-sP}](s,x,y)$ of the appearing exponential of $P$ is called the {\em heat kernel} because it satisfies the {\em heat equation}
$$(\pa_s + P_x)\left[e^{-sP}\right](s,x,y)=0\,.$$
It is possible to prove that, at least when $M$ is compact and $P\geq 0$ or, more weakly, it is bounded below 
 (but there are several other generalizations, see \cite{Wald79,Moretti2} also for references) the heat kernel exists, is an integrable function and, in geodesically convex neighborhoods, the following expansion is valid: 
\begin{align}\label{eq_DSExpansion}\left[e^{-sP}\right](s,x,y)&=\frac{1}{(4\pi s)^2}\;e^{-\frac{\si(x,y)}{2s}}\sum^\infty \limits_{n=0}a_n(x,y)s^n\;+\;\text{`smooth biscalar'}\\
&=\frac{1}{(4\pi s)^2}\;e^{-\frac{\si(x,y)}{2s}-m^2s}\sum^\infty \limits_{n=0}\alp_n(x,y)s^n\;+\;\text{`smooth biscalar'}\,,\notag\end{align}
The convergence of the series is  ``asymptotic'' in a certain  sense
discussed in sec. 1.2 of \cite{Moretti2}  (and it is not asymptotic in the Lorentzian case).
The functions $a_n(x,y)$ are smooth biscalars completely determined by the recursive differential equations \cite{FullingBook, Moretti2}
\beq\label{eq_DeWittSchwinger}a_0=u\,,\qquad P_xa_n+\sigma_{;}^{\mu}a_{n+1;\mu}+\left(\frac12 \Delta \si +n-1\right)a_{n+1}=0\eeq
with $u$, again, denoting the square root of the  VanVleck-Morette determinant and $\alp_n$ being related to $a_n$ via
\beq\label{eq_DeWittSchwinger2}\alp_n=\sum\limits^n_{j=0}\frac{\left(-m^{2}\right)^j}{j!}a_{n-j}\,.\eeq
In the Lorentzian framework, formally speaking, one gets the same
expansion and the same coefficients, though defined with the Lorentzian
metric and, in particular $\Delta$ is replaced by $\Box$ in the
recursive equations above.

The reason why we have displayed the  expansion (\ref{eq_DSExpansion}) in two versions is the following: the version in terms of $\alp_n$ is more convenient in the regularisation procedure, because the appearing $m^2$  avoids potential infrared singularities in the integral with respect to $s$ present in $L_\text{eff}$. In contrast, the version given in terms of $a_n$ is important to show the relation between the Hadamard coefficients $v_n$ and the $a_n$. Namely, a short computation and comparison with the scalar Hadamard recursion relations discussed in subsection \ref{HadamardCoefficients} reveals the well-known identity 
\beq\label{eq_HadamardDeWitt}a_{n+1}=(-)^{n+1}\;2^{n+1}\;n!\;v_n\,.\eeq

\section{The symmetry of the Seeley-DeWitt/Hadamard coefficients}
An intriguing issue concerns the symmetry of the above introduced functions $a_n(x,y)$
(or, equivalently, $v_n(x,y)$) under interchange of their arguments $x$ and $y$ when $(M,g)$ is smooth and either Riemannian or Lorentzian. The symmetry is by no means evident in view of the fact that their definition, relying on 
(\ref{eq_PDEVNScalar}), (\ref{eq_PDEV0Scalar})  and (\ref{eq_PDEUScalar}) is highly non-symmetric when swapping $x$ and $y$ because derivatives act on $x$ only. In \cite{Wald3}, such a symmetry was (indirectly) argued to hold for the analytic case. In \cite{FSW78},
the symmetry was argued to hold for the $C^\infty$ case and used to simplify some technical computations. 
Nevertheless, these papers did not report the corresponding
proof. In the literature concerning the point-splitting procedure
successive to  \cite{Wald3}, e.g. \cite{BO86}, this
symmetry is assumed to hold implicitly.
The issue was analysed in \cite{Moretti2,Moretti} and a proof was
presented both for the Euclidean and the Lorentzian case. In the
following we review the main ideas of the proof, and refer to \cite{Moretti2,Moretti} for technical details.

\begin{theorem} The Seeley-DeWitt/Hadamard coefficients $a_n(x,y)$ defined in a geodesically convex neighborhood of a $C^\infty$ manifold $(M,g)$, either Riemannian or Lorentzian, associated with the Klein-Gordon (resp. Riemannian or Lorentzian) operator $P$ (with $V$ smooth but arbitrary)  are symmetric if swapping  the arguments $x$ and $y$.
\end{theorem}

\noindent {\em Sketch of proof}. First of all, given a Riemannian manifold $M$ and a geodesically convex neighborhood ${\cal O} \subset M$, where the $a_n$ are defined, one modifies $M$ and $V$ outside ${\cal O}$, in order to produce a compact manifold $M' \supset {\cal O} $ equipped with an essentially self-adjoint  operator $P'$ bounded below and coinciding with $P$ in  ${\cal O}$. This does not affect the definition of the $a_n$ since they depend on the local geometry in ${\cal O}$  only.  
As $e^{-s\overline{P'}}$ is self-adjoint, it follows that $[e^{-s\overline{P'}}](s,x,y)= [e^{-s\overline{P'}}](s,y,x)$ for $x,y \in {\cal O}$. The convergence properties 
of the expansion (\ref{eq_DSExpansion}) \cite{Moretti2}  assure in
this case that $\partial^\alpha_x\partial^\beta_y
\left(a_n(x,y)-a_n(y,x)\right)|_{x=y}=0$ for any  
multiindices $\alpha,\beta$, referring to any fixed coordinate frame on ${\cal O}$.  If the metric $g$ is a real analytic function of the said coordinates, the coefficients $a_n(x,y)$ turn out to be analytic functions of the coordinates on ${\cal O} \times {\cal O}$ and, just due to the Taylor expansion, one has $a_n(x,y)= a_n(y,x)$ in a neighborhood of the diagonal of the open set ${\cal O} \times {\cal O}$ and thus everywhere thereon since it is  connected.  Whenever  the smooth metric $g$ is not real analytic, it is however possible to construct a class of analytic metrics $g_t$, depending on the parameter $t\in (0,1)$, uniformly converging to $g$ as $t\to 0$. In view of the structure of the equations  (\ref{eq_PDEVNScalar}), (\ref{eq_PDEV0Scalar})  and (\ref{eq_PDEUScalar}), exploiting a suitable version of the theorem on the analytic 
dependence of parameters for the solutions of a class of differential equations,
one can show \cite{Moretti2} that the associated coefficients
$a^{(t)}_n(x,y)$ satisfy as expected $a^{(t)}_n(x,y) \to a_n(x,y)$ as
$t\to 0$, assuring the symmetry of $a_n(x,y)$ since
$a^{(t)}_n(x,y)=a^{(t)}_n(y,x)$ if $t>0$.\\\indent Let us pass to the
case where $(M,g)$ is Lorentzian.  It is now enough to establish the
symmetry property  for the analytic case because the extension to the
smooth case is identical to that in the Euclidean case. So we assume
that in a geodesically convex neighborhood  ${\cal O}\subset M$  the
metric $g$ and thus, the coefficients $a_n(x,y)$, are real analytic
functions in a suitable 
coordinate frame thereon. As proved in \cite{Moretti}, it is possible to think of  ${\cal O}$ as a real section  of a complex manifold 
 ${\cal U}$ equipped with a complex analytic metric $h$ whose restriction to ${\cal O}$ is nothing but $g$. The pivotal fact is that this  complex manifold admits another real section ${\cal O}_E$ whose metric $g_E$ obtained by restriction of $h$ is Euclidean and ${\cal O}_E$ is geodesically convex as well.
 In this sense the real analytic Lorentzian  metric $g$ is the complex
 analytic continuation of the   real analytic Euclidean metric
 $g_E$. This procedure realizes a local Wick rotation, depending,
 however, on the chosen coordinate frame.
 Once again,  exploiting a suitable version of the theorem on the analytic 
dependence of parameters for the solutions of a class of differential equations (on complex manifolds now), one verifies  
that the function  $a_n(x,y)-a_n(y,x)$ in  ${\cal O}$ is the analytic continuation of a corresponding function
 $a^{(E)}_n(x_E,y_E)- a^{(E)}_n(y_E,x_E)$ in  ${\cal O}$. Since the latter is the zero function, as previously established, we have 
 $a_n(x,y)= a_n(y,x)$ in  ${\cal O}$ for the Lorentzian (analytic) case, too.

\section{The equivalence of the local $\zeta$-function and the Hadamard regularisation scheme in static spacetimes}
We intend to give a rigorous and well-known   meaning to the formal definition (\ref{stressSeff}) in the Euclidean case relying upon an analytic continuation procedure, another procedure  will be discussed later.
We start by providing a definite meaning \cite{Hawking} to the
determinant $\det P$ in the former identity (\ref{idtracedet}).  If $P$ is a $n\times n$ positive-definite Hermitian matrix with 
eigenvalues $0<\lambda_1\leq \lambda_2\leq \cdots \leq \lambda_n$ the following elementary result is valid
\begin{eqnarray}
\det P = e^{-\frac{d\zeta(s|P)}{ds}|_{s=0}}\quad \mbox{if}\quad \zeta(s|P) := \sum_{j=1}^n \lambda_j^{-s}\:, z \in \bC\:. \label{det'}
\end{eqnarray}
This trivial result can be generalized to  non-negative self-adjoint operators whose spectrum is discrete and each eigenspace has a finite dimension as it happens for the unique self-adjoint extension  $\overline{P}$ of Laplace-Beltrami operators in compact Riemannian manifolds or perturbations of the form $P = -\Delta + V$  in (\ref{PE}).  Consider
 the series with $s\in \bC$ (the prime on the 
sum henceforth  means that any possible null eigenvalues is omitted and, as is customary, we write $P$
in place of $\overline{P}$)
\begin{eqnarray}
\zeta(s| P) := {\sum_j}' \lambda^{-s}\:.
\label{sum0}
\end{eqnarray}
As is known, the series converges absolutely to an analytic function if $Re\, s$ is large enough.
As in (\ref{det'}), the idea \cite{Hawking} is to define, once again, if the right-hand side exists:
\begin{eqnarray}
\det P := e^{-\frac{d\zeta(s|P)}{ds}|_{s=0}}\quad \mbox{and}\quad 
S_\text{eff}: =- \frac{1}{2}\frac{d\zeta(s|P)}{ds}|_{s=0} \:. \label{Pdetzeta}
\end{eqnarray}
However, the function $\zeta$ on the right-hand side has to be understood as the {\em analytic continuation
of the function defined by the series} (\ref{sum0}) in its convergence domain, since the series may diverge at $s=0$ -- and this is the 
standard situation in the infinite-dimensional case!  For $P$ as in
(\ref{PE}) the procedure works and (\ref{Pdetzeta}) makes sense
(e.g. see \cite{MorettiZeta0}).

The idea can be further implemented giving rise to a $\zeta$-function  procedure to compute the stress energy tensor, suitably interpreting (\ref{stressSeff}).  As a matter of fact \cite{MorettiZeta}, one takes the $g_{\mu\nu}$ functional derivatives of the right hand side of (\ref{sum0}):
\beq
Z_{\mu\nu}(s,x|P) := {\sum_j}' \frac{2}{\sqrt{|\det g|}}\frac{\de \lambda^{-s}}{\de g^{\mu\nu}} \label{Z}
\eeq
and, after an analytic continuation of $Z_{\mu\nu}(s,x|P)$ in a neighborhood of $s=0$,  define
\beq \langle\wick{T_{\mu\nu}}(x)\rangle_\Om\mydef- \frac{1}{2}\frac{dZ_{\mu\nu}(s,x|P/\mu^2)}{ds}|_{s=0} \,.\label{stressSeffz}\eeq
The arbitrary mass scale $\mu$ has to be introduced due to dimensional reasons.
The outlined procedure works (the functional derivative can be
computed giving a precise meaning to the right-hand side of (\ref{Z})
and the series converges to an analytic function for $Re \,s$
sufficiently large) at least for compact Euclidean manifolds with $P$
as in (\ref{PE}) as proved in \cite{MorettiZeta}. In particular it can
be verified that $\langle\wick{T_{\mu\nu}}(x)\rangle_\Om$ defined as
in (\ref{stressSeffz}) is conserved and produces the conformal
anomaly. However the overall issue remains: how to relate
$\langle\wick{T_{\mu\nu}}(x)\rangle_\Om$ with the renormalized stress
energy tensor in Lorentzian spacetime and for a state $\Omega$? An
answer was presented in \cite{MorettiZeta} and \cite{Mo03} (see
especially theorem 2.2 in the latter)
\begin{theorem}
Let $(M, g)$ be a globally hyperbolic spacetime endowed with a global Killing time-like
vector field normal to a compact spacelike Cauchy surface and a Klein-Gordon operator $P$ in (\ref{evolution}) where $V$ does not depend on the Killing time. Consider a compact Euclidean section of
the spacetime $(M_E,g_E)$ obtained by:

  (a) a Wick analytic continuation with respect to Killing time, and 

(b) an identification of the Euclidean time into Killing orbits of period
$\beta >0$.\\
 Let $S$ be the unique Green function of the Euclidean Klein-Gordon operator
$P_E$ associated to $P$ by means of the Wick rotation
when $P_E$ is assumed to be strictly positive, and  $\Omega$ the Hadamard state associated to  $S$.\\
$\langle\wick{T_{\mu\nu}}(x)\rangle_\Om$ computed through  the
$\zeta$-function procedure (\ref{stressSeffz}) coincides with  
the analytic continuation of Killing time in static coordinates (Wick rotation) of $\langle\wick{T_{\mu\nu}}(x)\rangle_\Om$ computed through  the point-splitting prescription  as in (\ref{EMTensor}) when the scales $\lambda$ and $\mu$ are chosen suitably.
\end{theorem}

\section{The well-posedness of the DeWitt-Schwinger regularisation prescription in Lorentzian spacetimes}

As described in subsection \ref{SeeleyDeWitt}, one can formally obtain the expectation value of the regularised stress-energy tensor by first regularising the effective action and then taking the functional derivative with respect to the metric. If such regularisation of the effective action leads to a diffeomorphism-invariant action, the associated stress-tensor expectation value is automatically conserved \cite[app. E]{WaldBook}. The original motivation behind the DeWitt-Schwinger point-splitting regularisation was indeed to proceed in this way, whereby the effective action is regularised by means of the formal asymptotic expansion of the heat kernel. However the actual computations in the seminal works \cite{Christensen0, Christensen} have been performed directly on the level of the stress-tensor, without the detour via the effective action. Nevertheless, we shall now repeat the formal steps of this detour before we proceed to give a reformulation of the regularisation prescription used in \cite{Christensen0, Christensen} which is both rigorous on general smooth Lorentzian spacetimes and takes into account the state dependence of $\langle\wick{T_{\mu\nu}}\rangle_\Om$.

To wit, inserting the heat kernel expansion \eqref{eq_DSExpansion} in the definition of $L_\text{eff}$ \eqref{eq_Leff}, one finds that the contributions due to $n=0$, $n=1$, and $n=2$ lead to divergent integrals with respect to $s$ if $x=y$ (or if $x$ and $y$ are lightlike related in the Lorentzian case). Particularly, we observe UV-divergences at the lower integration limit, while infrared divergences at the upper limit are not present on account of $e^{-m^2s}$. One therefore defines the {\em renormalised effective Lagrangean} in the Riemannian case as
\beq\label{eq_Leffren}L_\text{eff, ren}(x)\mydef-\int\limits^\infty_0ds\;\frac{1}{s}\frac{1}{(4\pi s)^2}\;e^{-m^2s}\sum^\infty \limits_{n=3}\alp_n(x,x)s^n\,,\eeq
and $\langle\wick{T_{\mu\nu}}\rangle_\Om$ by the functional derivative of the associated {\em renormalised effective action}
$$S_\text{eff,ren}\mydef \int\limits_M d_gx\;L_\text{eff, ren}(x)\,.$$
In the Lorentzian case, one would define these quantities by replacing $s$ with $is$. We point out two important things. First, we know that the Hadamard coefficients $v_n$ and, hence, $\alp_n$ are covariant biscalars. Therefore, the renormalised effective Lagrangean is a scalar, and, consequently $\langle\wick{T_{\mu\nu}}\rangle_\Om$ is automatically conserved. Secondly, it is clear that $L_\text{eff, ren}(x)$ has {\em no state dependence whatsoever}. This holds because we know that the coefficients $\alp_n$ are completely specified by local curvature terms and $m$, $\xi$. In fact, we have ``lost'' the state dependence by disregarding the smooth remainder term in the expansion of the heat kernel. Defining $\langle\wick{T_{\mu\nu}}\rangle_\Om$ via $L_\text{eff, ren}(x)$ therefore completely disregards the state dependence of $\langle\wick{T_{\mu\nu}}\rangle_\Om$.
Apart from the appearance of the Hadamard coefficients, there does not seem to be a close relation to our definition of $\langle\wick{T_{\mu\nu}}\rangle_\Om$ in terms of applying a suitable bidifferential operator $D_{\mu\nu}$ to the regularised two-point function and then taking the coinciding point limit. However, one can reformulate the above renormalisation of the effective action in the following way \cite{Christensen0, Wald79}. One formally pulls the functional derivative with respect to the metric in the definition of $\langle\wick{T_{\mu\nu}}\rangle_\Om$ under the integral with respect to $s$ and then finds via additional formal steps
$$\langle\wick{T_{\mu\nu}}\rangle_\Om\mydef\frac{2}{\sqrt{|\det g|}}\frac{\de S_\text{eff}}{\de g^{\mu\nu}}=\left[\frac{2}{\sqrt{|\det g|}}\frac{\de \sqrt{|\det g|} P(x,y)}{\de g^{\mu\nu}}\int\limits_0^\infty ds\;e^{-sP}\right]$$$$=\left[\frac{2}{\sqrt{|\det g|}}\frac{\de S}{\de g^{\mu\nu}\de \phi(x)\de \phi(y)}\int\limits_0^\infty ds\;e^{-sP}\right]=\left[\frac{\de T_{\mu\nu}}{\de\phi(x)\de\phi(y)}\int\limits_0^\infty ds\;e^{-sP}\right]$$$$=\left[D^\text{can}_{\mu\nu}(x,y)\left[P^{-1}\right](x,y)\right]\,,$$
where the outer square brackets denote the coinciding point limit. In the above formal derivation, it has been used that the integral kernel $P(x,y)=\de(x,y)P_x$ of $P$ is obtained as the second functional derivative of the classical action $S$ with respect to the field $\phi$ and that the canonical differential operator $D^\text{can}_{\mu\nu}$ we have considered in section \ref{HadamardCoefficients} is nothing but the second functional derivative of the classical stress-energy tensor with respect to the field $\phi$. In the context of renormalisation of the effective action in Lorentzian spacetimes, one usually considers $P^{-1}$ to be the {\em Feynman propagator} $\De_F$ (note that $P^{-1}$ is not unique). Hence, the divergences of $\langle\wick{T_{\mu\nu}}\rangle_\Om$ computed as above are interpreted to stem from the divergences of the Feynman propagator at coinciding points. To renormalise $\langle\wick{T_{\mu\nu}}\rangle_\Om$ in the Lorentzian case, one therefore inserts the DeWitt-Schwinger expansion \eqref{eq_DSExpansion} of the heat kernel in the integral expression for $P^{-1}$ in terms of $e^{-sP}$ to obtain
$$\De_F(x,y)\mydef \lim_{\ep\downarrow 0}\int\limits^\infty_0ds\;\frac{1}{(4\pi s)^2}\;e^{-\frac{\si(x,y)+i\ep}{2s}-m^2s}\sum^\infty \limits_{n=0}\alp_n(x,y)s^n$$$$=\lim_{\ep\downarrow 0}\;\frac{1}{8\pi^2}\;\sum\limits^\infty_{n=0}\left(\frac{\si+i\ep}{2m^2}\right)^{\frac{n-1}{2}}K_{n-1}\left(\sqrt{2m^2(\si+i\ep)}\right)\;\alp_n(x,y)\,,$$
where the $\ep$-prescription suitable for the Feynman propagator has been inserted and an integral identity for the {\em modified Hankel function} $K_n$ has been used. Expanding this in powers of $\si$, inserting $\alp_0=u$, and removing the $\ep$-prescription from the regular terms, we find
$$\De_F(x,y)=\frac{1}{8\pi^2}\left\{\frac{u}{\si+i\ep}+\log\left(\frac{(\si+i\ep)m^2e^{2\ga}}{2}\right)\left(\frac{m^2 u}{2}-\frac{\alp_1}{2}+\frac{m^4 u\si}{8}+\frac{\alp_2 \si}{4}-\frac{m^2\alp_1 \si}{4}\right)\right.$$
$$\left. -\frac{m^2 u}{2} -\frac{5 m^4 u \si}{16}+\frac{\alp_1 \si}{2}-\frac{\alp_2 \si}{4}+\frac{\alp_2}{2m^2}+O\left(\si^2\log(\si+i\ep)\right)\right\}\,,$$
and, inserting the relation between $\alp_n$ and $v_n$ as given in \eqref{eq_HadamardDeWitt} and \eqref{eq_DeWittSchwinger2}, we see explicitly that, barring the different $\ep$-prescription, $\De_F(x,y)$ displays the Hadamard singularity structure. Again we point out that the `correct' Feynman propagator is always {\em state-dependent}, while the above expression is manifestly {\em state-independent}, being essentially only a local curvature expression. Once more, this stems from the fact that one has disregarded the smooth (non-local) remainder in the expansion of the heat kernel. Note however, that, while this smooth remainder term is essentially well-understood in the Riemannian case, this does not hold in the Lorentzian case, as there the whole DeWitt-Schwinger expansion is not rigorous. Hence, in contract to the Hadamard expansion of the two-point function of a Hadamard state, there does not seem to be a possibility to introduce the state-dependence in a rigorous way in the DeWitt-Schwinger renormalisation as we have presented it up to now. Moreover, it is already visible from the few terms we have provided that, in term of the Hadamard series, the smooth term $w$ of the above distribution contains inverse powers of the mass and, hence, diverges in the massless limit. Of course this is particularly also the case for the logarithmic terms, which displays the infrared singularity of the integrals with respect to $s$.

By the well known distributional identities (see for instance \cite{Reed, FullingBook})

$$\lim_{\ep\downarrow 0}\;\frac{1}{x+i\ep}=\cP \frac{1}{x}+i\pi \de(x)\,,\qquad \lim_{\ep\downarrow 0}\;\log(x+i\ep)=\log |x|+\pi i\Th(-x)\,,$$
where $\cP$ denotes the principal value, one finds that $[D^\text{can}_{\mu\nu}(x,y)[P^{-1}](x,y)]$ is a complex number, and one therefore has to consider its real-part as the `correct' definition of $\langle\wick{T_{\mu\nu}}\rangle_\Om$. In terms of the Hadamard series we have discussed in subsection \ref{HadamardCoefficients}, this corresponds to consider the symmetric part
$$h^s(x,y)\mydef\frac12\left(h(x,y)+h(y,x)\right)$$ in the definition
of the point-splitting prescription, where here and in the following
we omit the $\ep$ in $h_\ep$. Note that this encodes the same information as the full $h$ in the coinciding point limit.

With the setup we have just described, the regularisation of  $\langle\wick{T_{\mu\nu}}\rangle_\Om$ goes as follows \cite{Christensen0}. Applying $D^\text{can}_{\mu\nu}$ to the real part of $\De_F$ given in terms of the DeWitt-Schwinger expansion, one finds that, as in the regularisation of the effective action, divergences come from the terms of order $n=0$, $n=1$, and $n=2$. Hence, one identifies the divergent part of the stress-energy tensor as $D^\text{can}_{\mu\nu}$ applied to the first three terms in the expansion. However, these terms of course also contain smooth contributions, and in fact one has to take care to not subtract `too much', otherwise one could spoil the covariant conservation of $\langle\wick{T_{\mu\nu}}\rangle_\Om$ which was automatic in the renormalisation of the effective action. Namely, although in the latter renormalisation one has subtracted the first three terms of the DeWitt-Schwinger expansion as well, some of the subtractions have been zero on account of the vanishing of $\si$ in the coinciding point limit. As we now derive the series, we could accidentally introduce these vanishing terms since second derivatives of $\si$ do not vanish in the coinciding point limit present in $L_\text{eff}(x)$. However, as observed in \cite{Christensen0}, this can be avoided if one defines the divergent part of $\langle\wick{T_{\mu\nu}}\rangle_\Om$ by applying $D^\text{can}_{\mu\nu}$ to the real part of $\De_F(x,y)$ and then discarding all terms which are proportional to inverse powers of the mass. Proceeding like this, Christensen has computed in \cite{Christensen0, Christensen} the divergent part of the quantum stress-energy tensor and has obtained the trace anomaly by computing the negative trace of the divergent part. This follows because the full expression of the stress-energy tensor must have vanishing trace in the conformally invariant case as, despite its divergence, it is completely given in terms of the real part of $\De_F$, which in turn is a bisolution of the Klein-Gordon equation by its very construction in terms of the DeWitt-Schwinger series. Note that Christensen has introduced the massless limit by replacing the $m^2$ in the logarithmic divergence by an arbitrary scale $\la$. This may seem rather ad-hoc, but from our point of view this is a reasonable procedure if we remember that we have defined the Hadamard series with an arbitrary length scale in the logarithm right from the start; moreover, as shown in \cite{Wald79}, the scale can be consistently introduced as an ``IR-regulator''. Finally, since the smooth terms proportional to inverse powers of the mass had been discarded to assure conservation, the massless limit could be performed in a meaningful way.

Having reviewed the DeWitt-Schwinger point-splitting renormalisation of the stress-energy tensor, let us recapitulate its seeming disadvantages.

\begin{itemize}
\item[a)] It has been defined via an expansion of the heat kernel which is not well-defined in general curved Lorentzian spacetimes.
\item[b)] It does not take into account the state dependence of $\langle\wick{T_{\mu\nu}}\rangle_\Om$.
\item[c)] It employs a Hadamard series whose smooth part $w$ diverges in the massless limit.
\end{itemize}

\noindent We shall now give a regularisation prescription which closely mimics the one of Christensen, but disposes of the above three problems.

\begin{theorem}
 \label{thm_DeWittSchwinger} Let $\Om_2(x,y)\mydef\langle\phi(x)\phi(y)\rangle_\Om$ be the two point function of a Hadamard state $\Om$, let
$$\Om^s_2(x,y)\mydef \frac12(\Om_2(x,y)+\Om_2(y,x))\,,$$
and let us define for $x$ and $y$ in a geodesically convex neighbourhood
$$h_\text{DS}(x,y)\mydef\frac{1}{8\pi^2}\left\{\cP\frac{u}{\si}+\log\left|\frac{\si m^2e^{2\ga}}{2}\right|\left(\frac{m^2 u}{2}-\frac{\alp_1}{2}+\frac{m^4 u\si}{8}+\frac{\alp_2 \si}{4}-\frac{m^2\alp_1 \si}{4}\right)\right.$$
$$\left. -\frac{m^2 u}{2} -\frac{5 m^4 u \si}{16}+\frac{\alp_1 \si}{2}-\frac{\alp_2 \si}{4}+\frac{\alp_2}{2m^2}\right\}\,,$$
$$\mydef h_0(x,y)+h_m(x,y)\,,$$
$$h_m(x,y)\mydef \frac{1}{8\pi^2}\frac{\alp_2}{2m^2}\,,\quad h_0(x,y)\mydef h_\text{DS}(x,y)-h_m(x,y)\,.$$
Moreover, let us split the canonical bidifferential operator
$$D^\text{KG,can}_{\mu\nu}=(1-2\xi)g^{\nu^\prime}_\nu\nabla_{\mu}\nabla_{\nu^\prime}-2\xi\nabla_{\mu}\nabla_\nu+\xi G_{\mu\nu}+g_{\mu\nu}\left\{2\xi \square_x+\left(2\xi-\frac12\right)g^{\rho^\prime}_\rho\nabla^\rho\nabla_{\rho^\prime}-\frac12 m^2\right\}$$
as
$$D^\text{KG,can}_{\mu\nu}\mydef D^0_{\mu\nu}+D^m_{\mu\nu}\,,\qquad D^m_{\mu\nu}\mydef -\frac12 m^2g_{\mu\nu}\,,\qquad D^0_{\mu\nu}\mydef D^\text{KG,can}_{\mu\nu}-D^m_{\mu\nu}\,.$$
The stress-energy tensor regularisation prescription defined as
$$\langle\wick{T^\text{DS}_{\mu\nu}}\rangle_\Om\mydef \left[D^\text{KG,can}_{\mu\nu}\left(\Om^s_2-h_\text{DS}\right)+D^0_{\mu\nu}h_m\right]$$
fulfils the Wald axioms of state-independence, local covariance and covariant conservation. Particularly, it displays the trace anomaly
$$\left.g^{\mu\nu}\langle\wick{T^\text{DS}_{\mu\nu}}\rangle_\Om\right|_{m^2=0}=-\frac{[v_1]}{4\pi^2}\,.$$
\end{theorem}

\begin{proof}
First of all, let us remark that the regularisation prescription is well-defined, as the relation between $\alp_n$ and $v_n$ given in \eqref{eq_HadamardDeWitt} and \eqref{eq_DeWittSchwinger2} implies that $\Om^s_2-h_\text{DS}$ is of class $C^2$ (the worst terms in $\Om^s_2-h_\text{DS}$ are of the form $\si^2\log \si$). Additionally, the prescription fulfils the requirement of local covariance, since it only involves the subtraction of objects given in terms of the DeWitt-Schwinger/Hadamard coefficients. Moreover, state-independence follows manifestly from the definition, as the modification of the canonical prescription is given in terms of $D^0_{\mu\nu}h_m$, which is a state-independent term.

To prove covariant conservation, we recall that, in the proof of theorem \ref{thm_TensorScalar}, it has been implicitly computed that for any smooth  biscalar $B(x,y)$ the following relation holds
$$\nabla^\mu\left[D^\text{KG,can}_{\mu\nu}\;B\right]=-\left[\nabla_{\nu^\prime}P_x\;B\right]\,.$$
Applying this to our current case, we find
$$\nabla^\mu\langle\wick{T^\text{DS}_{\mu\nu}}\rangle_\Om=\nabla^\mu\left[D^\text{KG,can}_{\mu\nu}\left(\Om_2-h_\text{DS}\right)+D^0_{\mu\nu}h_m\right]$$
$$=-\left[\nabla_{\nu^\prime}P_x\left(\Om^s_2-h_\text{DS}\right)+\nabla_{\nu^\prime}P^0_x \;h_m\right]=-\left[\nabla_{\nu^\prime}\left(P_x\Om^s_2-P_x h_\text{DS}+P^0_x h_m\right)\right]\,,$$
where we have defined
$$P^0\mydef -\Box+\xi R\,.$$
A straightforward computation employing the Hadamard/DeWitt-Schwinger recursion relations yields
$$P h_\text{DS}=P^0h_m\,.$$
From this and the fact that $\Om^s_2$ naturally solves the Klein-Gordon equation, conservation follows.

By Wald's results \cite{Wald3}, the above findings already imply that $\langle\wick{T^\text{DS}_{\mu\nu}}\rangle_\Om$ displays the ``correct'' trace anomaly. However, it is instructive to compute it explicitly. To this end, we obtain with steps similar to the ones already taken and using the implicit computational results obtained in the proof of theorem \ref{thm_TensorScalar}
$$g^{\mu\nu}\langle\wick{T^\text{DS}_{\mu\nu}}\rangle_\Om=g^{\mu\nu}\left[D^\text{KG,can}_{\mu\nu}\left(\Om_2-h_\text{DS}\right)+D^0_{\mu\nu}h_m\right]$$
$$=-\left[\left(P_x-m^2\right)\left(\Om^s_2-h_\text{DS}\right)+P^0_x h_m\right]=m^2\left[\Om^s_2-h_\text{DS}\right]=m^2\left[\Om^s_2-h_0\right]-\frac{1}{16\pi^2}\left[\alp_2\right]$$
$$=m^2\left[\Om^s_2-h_0\right]-\frac{1}{8\pi^2}\left[2v_1-m^2v_0+\frac{m^4u}{4}\right]\,.$$
\end{proof}

\noindent A few comments on the result are in order. First, on practical grounds, the above result is really equivalent to the computation of Christensen in \cite{Christensen0, Christensen} because the terms of the DeWitt-Schwinger expansion we have omitted are all proportional to $\si^2$ and, hence, vanish upon application of the occurring differential operators in the coinciding point limit. In this sense, we have been able to put his results on firm grounds. Secondly, the modification term $D^0_{\mu\nu}h_m$ does not ``simply cure the conservation anomaly''. In this sense, the regularisation prescription just analysed differs from the one introduced in \cite{Wald3} and improved in \cite{Mo03} in that it assures conservation in a different way. Namely, conservation does not follow by adding a term by hand or by modifying the classical stress energy tensor. In contrast, it follows by the explicit structure of the smooth term $w$ in the DeWitt-Schwinger two-point function in combination with discarding specific terms proportional to inverse powers of the mass. Recall that the latter procedure has been motivated by analysing carefully the subtractions in the (non-rigorous) regularisation of the effective action. Finally, let us remark that the above prescription yields a smooth $\langle\wick{T^\text{DS}_{\mu\nu}}\rangle_\Om$, although $\Om^s_2-h_\text{DS}$ is only known to be twice-differentiable, because the non-smooth terms proportional to $\log \si$ all vanish in the coinciding point limit.

\section{Conclusions}

In this work we have proven parts of Wald's conjecture in
\cite{Wald79}, namely, that the local $\zeta$-function and the
DeWitt-Schwinger regularisation prescriptions of the stress-energy
tensor for quantum fields in curved spacetimes both satisfy the
essential Wald axioms of state-independence, local covariance and
covariant conservation, albeit only in static Lorentzian manifolds in the $\zeta$-function case. Moreover, we have shown how to reformulate the DeWitt-Schwinger prescription in such a way that it is able to take into account the full state dependence of the stress-energy tensor expectation value. Consequently, both prescriptions are physically meaningful for all we know and the results obtained by using them can be trusted although they do not seem to be rigorous at first glance. 

As a side note, it turns out that both the DeWitt-Schwinger and the Hadamard prescriptions require an equal amount of computation work in contrast to the seemingly widespread attitude that the latter is rigorous whereas the former is better suited for computations (see e.g. \cite{Decanini1, Decanini2}). Notwithstanding, the power of the Hadamard prescription is that it can be directly generalised to obtain a regularisation of {\it all} field polynomials \cite{BFK, BF00, HW01, Mo03, HW04}; this is done by means of the mathematical tools of ``microlocal analysis'' \cite{Hormander0, Hormander02, Hormander}, which have been introduced to the framework of quantum field theory in curved spacetimes in \cite{Radzikowski, Radzikowski2}. It is not clear how this can be done with the DeWitt-Schwinger prescription.

Finally, we would like to mention a consistency problem of the semiclassical Einstein equation which is not solved in full generality by any regularisation prescription -- this equation is only well-defined in general if one can associate to a each spacetime $M$ a ``unique'' state $\Om_M$, because the spacetime is unknown prior to solving the semiclassical Einstein equation, but solving the equation is not possible before saying which $\Om$ is chosen for the evaluation of $\langle\wick{T_{\mu\nu}}\rangle_\Om$. However, it is known that such an assignment of a state to a spacetime is impossible \cite{HW01, BFV, Fewster}. Notwithstanding, this problem can be overcome if it is possible to restrict oneself to a limited class of spacetimes on which a unique state can be defined, this has been done in \cite{Nicola}.

\end{document}